%% file: main.tex
\documentclass[runningheads]{llncs}

\RequirePackage{amsmath}
\RequirePackage{amssymb}
\RequirePackage[colorlinks,citecolor=blue,urlcolor=blue]{hyperref}
\RequirePackage[utf8]{inputenc}
\RequirePackage{xspace}
\RequirePackage{comment}
\RequirePackage{tikz}
\RequirePackage{subfigure}
\RequirePackage[misc,geometry]{ifsym}

\tikzstyle{index on}=[inner sep=2pt, white, circle, fill=black]
\tikzstyle{index off}=[inner sep=2pt, black, circle, draw]
\tikzstyle{index gray}=[inner sep=2pt, black, circle, fill=lightgray,minimum size=0.85cm]
\tikzstyle{index blue}=[inner sep=2pt, black, circle, fill=myblue,minimum size=0.85cm]
\tikzstyle{index white}=[inner sep=2pt, black, circle, fill=white]
\tikzstyle{opaque}=[fill=gray,fill opacity=.1]
\tikzstyle{emptydraw}=[draw opacity=0]
\tikzstyle{counter}=[densely dashed]

\newcommand{\dep}[2]{=\hspace{-0.1cm}(#1,#2)}
\newcommand{\con}[1]{=\hspace{-0.1cm}(#1)}
\newcommand{\NE}{\normalfont\textsc{ne}}
\newcommand{\PL}{\ensuremath{\mathrm{PL}}\xspace}
\newcommand{\PLNE}{\ensuremath{\PL(\NE)}\xspace}
\newcommand{\DEP}{\ensuremath{\PL(\con{\cdot})}\xspace}

\newcommand{\Prop}{\ensuremath{\mathsf{Prop}}\xspace}
\newcommand{\VAL}[1]{\ensuremath{\mathrm{VAL}[#1]}}

\newcommand{\SAT}[1]{\ensuremath{\mathrm{SAT}[#1]}}
\newcommand{\MC}[1]{\ensuremath{\mathrm{MC}[#1]}}

\renewcommand{\phi}{\varphi}

\makeatletter
\newcommand{\leqnomode}{\tagsleft@true}
\newcommand{\reqnomode}{\tagsleft@false}
\makeatother

\title{Complexity Results in Team Semantics: Nonemptiness Is Not So Complex}

\author{Aleksi Anttila\inst{1}\orcidID{0009-0007-3497-5411}\textsuperscript{(\Letter)} \and
Juha Kontinen \inst{1}\orcidID{0000-0003-0115-5154} \and
Fan Yang \inst{2}\orcidID{0000-0003-0392-6522}}
\institute{Department of Mathematics and Statistics, University of Helsinki, Finland\\\email{aleksi.ilari.anttila@gmail.com}\\\email{juha.kontinen@helsinki.fi} \and Department of Philosophy and Religious Studies, Utrecht University, the Netherlands\\\email{fan.yang.c@gmail.com}}
\authorrunning{A. Anttila et al.}

\date{}

\begin{document}

\maketitle


\begin{abstract}
We initiate the study of the complexity-theoretic properties of convex logics in team semantics. We focus on the extension of classical propositional logic with the nonemptiness atom $\NE$, a logic known to be both convex and union closed. We show that the satisfiability problem for this logic is $\mathrm{NP}$-complete, that its validity problem is $\mathrm{coNP}$-complete, and that its model-checking problem is in $\mathrm{P}$.
\end{abstract}

\setcounter{footnote}{0}
\section{Introduction} \label{cPLNE:section:introduction}

Team semantics is a mathematical framework for studying concepts and phenomena that arise in contexts involving a plurality of data, typically represented by sets of assignments or probability distributions. Logics employing team semantics have found application in many fields, including database theory \cite{Kontinen:2013:independence,Hannula:2020:polyteam}, Bayesian networks and probabilistic dependencies \cite{Corander:2019,Hirvonen24}, quantum foundations \cite{Durand:2018:probabilistic,Abramsky:2021:team,AlbertG22}, as well as philosophical logic and formal semantics \cite{inqsembook,ciardellibook,hawke,aloni2022}.

The core idea of team semantics is that formulas are interpreted over sets of evaluation points (called \emph{teams}) rather than single evaluation points as in standard Tarskian semantics. Depending on the setting, teams can be, for instance, sets of assignments, propositional valuations, or possible worlds. The origins of team semantics can be traced back to Hodges' \emph{trump semantics} \cite{hodges1997} for Hintikka and Sandu's \emph{independence-friendly logic} \cite{hintikka1989,hintikka1996}; Väänänen developed team semantics on the basis of Hodges' work to interpret \emph{dependence logic} \cite{vaananen2007}, his refinement of independence-friendly logic. Dependence logic extends classical first-order logic with \emph{dependence atoms} such as $\dep{x}{y}$; informally, $\dep{x}{y}$ expresses that the values of the variable $x$ functionally determine the values of $y$ within a team. Other team-based logics follow a similar pattern: a non-team-based logic is extended with connectives whose semantics make essential use of the fact the evaluation point is a team: a set, rather than a singleton. For instance, \emph{propositional inclusion logic} $\PL(\subseteq)$ \cite{hella2014,yang2022}---the propositional version of \emph{inclusion logic} \cite{galliani2012}, another of the most prominent team-based logics---extends classical propositional logic $\PL$ with \emph{inclusion atoms} such as $\alpha \subseteq \beta$, with $\alpha \subseteq \beta$ expressing that whatever truth values the $\PL$-formula $\alpha$ is assigned by some valuation in a team must also be assigned to $\beta$ by some valuation in the team.

Since the introduction of dependence logic, the expressivity and computational complexity of logics within the framework of team semantics have been extensively investigated (see, e.g., \cite{lohmann2013,hannula2018,Durand:2022:tractability,DurandKV24}). We contribute to this literature by characterizing the computational complexity of the satisfiability (SAT), validity (VAL), and model checking (MC) problems of the logic \PLNE \cite{vaananen2014,yang2017}, which extends classical propositional logic \PL with the \emph{nonemptiness atom} $\NE$, true in a team just in case the team is nonempty. Despite its simplicity, this logic has recently been attracting interest because it serves as the propositional basis for Aloni's \emph{bilateral state-based modal logic} BSML \cite{aloni2022,aloni2023}, a modal team-based logic which can be used to account for free-choice inferences and related linguistic phenomena. We show that $\SAT{\PLNE}$ is NP-complete; that $\VAL{\PLNE}$ is coNP-complete; and that $\MC{\PLNE}$ is in P.

The logic $\PLNE$ is also interesting due to its team-semantic \emph{closure properties}. Unlike dependence logic, $\PLNE$ is not \emph{downward closed} (where a logic is downward closed if for each of its formulas $\phi$, the truth of $\phi$ in a team $t$ implies truth in all subteams of $t$). Like inclusion logic, it is \emph{union closed} (the truth of a formula in a nonempty collection of teams implies its truth in the union of the collection); however, $\PLNE$ differs from inclusion logic in that it does not have the \emph{empty team property} (a logic has this property if all formulas are true in the empty team) and in that it is \emph{convex} (the truth of a formula in two teams implies its truth in all teams between the two with respect to the subset relation). The convexity property has not been studied in the literature on team semantics until recently \cite{Hella:2024,ciardelli2024,anttila2025convexteamlogics}, and, to our knowledge, the literature does not contain any results on the complexity-theoretic properties of logics which are convex but not downward closed (convexity can be seen as a generalization of downward closure to a setting in which the empty team property fails).

\begin{table}[h]
    \centering
\begin{tabular}{c|c|c|c} 
\hline
    Logic & SAT & VAL & MC \\
    \hline
    $\PL$  & NP \cite{cook1971,levin1973} & coNP \cite{cook1971,levin1973} & NC$^1$ \cite{Buss1987TheBF}\\
    $\PLNE$ & NP [Corollary \ref{cPLNE:coro:SAT(PLNE)_NP_complete}] & coNP [Prop. \ref{cPLNE:prop:VAL(PLNE)_coNP_complete}] & $\in\mathrm{P}$ [Prop. \ref{cPLNE:prop:MC(PLNE)_in_P}]\\
    $\DEP$  & NP \cite{lohmann2013} & NEXP \cite{virtema2017} & NP \cite{ebbing2012}\\
    $\PL(\subseteq)$ & EXP \cite{hella20192} & coNP \cite{hannula2018} & P \cite{hella2019}\\
    $\PL(\sim)$ & ATIME-ALT(exp, poly) \cite{hannula2018} & ATIME-ALT(exp, poly) \cite{hannula2018} & PSPACE \cite{muller2014}\\
    \hline
\end{tabular}
    \caption{Complexity results for some propositional team-based logics (completeness results unless stated otherwise).}
    \label{cPLNE:table:results}
\end{table}

Table \ref{cPLNE:table:results} lists our results alongside those for classical propositional logic $\PL$, propositional dependence logic $\DEP$ \cite{yangvaananen2016}, propositional inclusion logic ${\PL(\subseteq)}$, and \emph{propositional team logic} $\PL(\sim)$ \cite{vaananen2007,muller2014,luck2020}. The logic $\PL(\sim)$ is expressively complete for all propositional team properties \cite{kontinen2014} and many of the common team-semantic connectives are definable in $\PL(\sim)$, so it provides a natural upper bound for the complexity of team-based extensions of $\PL$. The proof in \cite{lohmann2013} that $\SAT{\DEP}$ is NP-complete---i.e., that adding dependence atoms to $\PL$ does not increase the complexity of the satisfiability problem---makes use of downward closure. Similarly, the proof in \cite{hannula2018} that $\VAL{\PL(\subseteq)}$ is coNP-complete---i.e., that adding inclusion atoms to $\PL$ does not increase the complexity of the validity problem---makes use of union closure; as does the proof of the P-completeness of $\MC{\PL(\subseteq)}$ in \cite{hella2019}. The complexity results for the logic $\PLNE$ follow this pattern: being convex (downward closed up to a point), the complexity of its satisfiability problem is the same as that of $\PL$; being union closed, the complexity of its validity problem is the same as that of $\PL$, and its model-checking problem is tractable (our proofs of these results also make use of the relevant closure properties). As a result, we see that, at least with regard to the satisfiability and validity problems, adding the nonemptiness atom $\NE$ to $\PL$ does not increase complexity---in contrast, and as might be inferred from Table  \ref{cPLNE:table:results}, many other team-based atoms and connectives \emph{do} increase complexity, often quite considerably.

The paper is structured as follows. In Section \ref{cPLNE:section:preliminaries}, we define the syntax and semantics of \PLNE and list some basic properties of this logic. In Section \ref{cPLNE:section:MC}, we treat the model-checking problem; in Section \ref{cPLNE:section:SAT}, the satisfiability problem; and in Section \ref{cPLNE:section:VAL}, the validity problem. Section \ref{cPLNE:section:conclusion} concludes with some discussion of potential future work.

\section{Preliminaries} \label{cPLNE:section:preliminaries}
In this section, we define the syntax and team semantics of classical propositional logic (\PL) and its extension \PLNE by the nonemptiness atom $\NE$.

We fix a (countably infinite) set $\Prop$ of propositional variables.
\begin{definition}[Syntax] \label{cPLNE:def:syntax}
The set of formulas $\alpha$ of \emph{classical propositional logic} \PL is generated by:
\[\alpha::=p \mid\neg p \mid \bot \mid \top \mid \alpha\wedge\alpha\mid\alpha\vee\alpha\]
where $p\in \Prop$. The language of \PLNE extends that of \PL with the \emph{nonemptiness atom} $\NE$.\footnote{Some formulations of the syntax of \PLNE in the literature feature a negation operator $\lnot$ that can be applied to all $\PL$-formulas, with the semantics $t\models \lnot\alpha $ iff for all $v\in t$, $\{v\}\not\models \alpha$. Our results also hold for this variant of \PLNE---this follows from the fact that formulas in this variant can be converted into negation normal form in polynomial time.}
\end{definition}

We let $\bigvee \emptyset:=\bot$, and $\bigwedge \emptyset:=\top$. We also call formulas of $\PL$ \emph{classical formulas}. The metavariable $\alpha$ ranges over classical formulas; and $\Lambda$, over sets of classical formulas; metavariables denoted by other Greek letters range over all (sets of) $\PLNE$-formulas.  We write $\mathsf{P}(\phi)$ for the set of propositional variables appearing in $\phi$. We write $\phi (\psi/\chi)$ for the result of replacing all occurrences of $\chi$ in $\phi$ by $\psi$. We call formulas of the form $p$, $\lnot p$, $\bot$, $\top$, or $\NE$ \emph{literals}.

Let $\mathsf{X}\subseteq \Prop$. A \emph{team with domain $\mathsf{X}$} is a set $t\subseteq 2^\mathsf{X}$ of valuations over $\mathsf{X}$. 

\begin{definition}[Semantics] \label{cPLNE:def:semantics}
For any formula $\phi$ and any team $t$ with domain containing $ \mathsf{P}(\phi)$, the satisfaction relation $t\models\phi$ is defined inductively by:
\begin{align*}
    &t\models p &&\text{iff} &&\text{for all $v\in t$, $ v(p)=1$;}\\
    &t\models \neg p &&\text{iff} &&\text{for all $v\in t$, $ v(p)=0$;}\\
    &t\models \bot &&\text{iff} &&t=\emptyset;\\
    &t\models \top && &&\text{always;}\\
    &t\models\phi\wedge\psi &&\text{iff} &&\text{$t\models\phi$ and $t\models\psi$;}\\
    &t\models\phi\vee\psi &&\text{iff} &&\text{there exist $s,u\subseteq t$ such that $t=s\cup u$, $s\models\phi$, and $u\models\psi$;}\\
    &t\models\NE &&\text{iff} &&\text{$t\neq\emptyset$.}
\end{align*}
For any set $\Gamma\cup\{\phi\}$ of formulas, we write $t\models \Gamma$ if $t\models\psi$ for all $\psi\in\Gamma$, and we write $\Gamma\models\phi$ if $t\models \Gamma$ implies $t\models \phi$. We write $\phi\models \psi$ for $\{\phi\}\models \psi$. We write $\models\phi$ if $t\models\phi$ for all teams $t$ with domain containing $ \mathsf{P}(\phi)$.
\end{definition}

\begin{figure}[t]
  \centering
  \subfigure[{$t \models p \vee \lnot p$}]{
  \label{fig:a}
    \begin{tikzpicture}[>=latex,scale=.8]
    
    \draw[opaque,rounded corners] (.2,1.8) rectangle (1.8, .2);
    
    \draw (-1,1) node[index gray, minimum size=0.8cm] (wp) {$v_{pq}$};
    \draw (1,1) node[index gray, minimum size=0.8cm] (wpq) {$v_{p\overline{q}}$};
    \draw (-1,-1) node[index gray, minimum size=0.8cm] (wq) {$v_{\overline{p}q}$};
    \draw (1,-1) node[index gray, minimum size=0.8cm] (w4) {$v_{\overline{pq}}$};
  	\end{tikzpicture}
  } 
  \hspace{.5cm}
 \subfigure[$s \models (p\land \NE) \vee (\lnot p \land \NE)$]{
   \label{fig:b}
   \hspace{0.26cm}
    \begin{tikzpicture}[>=latex,scale=.8]
    
    \draw[opaque,rounded corners] (-1.8,1.8) rectangle (-.2, -1.8);
    
    \draw (-1,1) node[index gray, minimum size=0.8cm] (wp) {$v_{pq}$};
    \draw (1,1) node[index gray, minimum size=0.8cm] (wpq) {$v_{p\overline{q}}$};
    \draw (-1,-1) node[index gray, minimum size=0.8cm] (wq) {$v_{\overline{p}q}$};
    \draw (1,-1) node[index gray, minimum size=0.8cm] (w4) {$v_{\overline{pq}}$};
  	
  	\end{tikzpicture}
      \hspace{0.26cm}
  }
    \hspace{.5cm}
 \subfigure[$u \models ((p\land \NE) \vee (c \land \NE))\land \NE$ \; $u \models \Diamond p \land \Diamond c$]{
   \hspace{0.5cm}
   \label{fig:c}
    \begin{tikzpicture}[>=latex,scale=.8]
    
    \draw[opaque,rounded corners] (-1.8,-1.2) -- (-1.8,-0.8) -- (0.8,1.8) --  (1.8, 1.8) -- (1.8,0.8) -- (-0.8,-1.8) -- (-1.8,-1.8) -- (-1.8,-1.2);
    
    \draw (-1,1) node[index gray, minimum size=0.8cm] (wp) {$v_{pc}$};
    \draw (1,1) node[index gray, minimum size=0.8cm] (wpq) {$v_{p\overline{c}}$};
    \draw (-1,-1) node[index gray, minimum size=0.8cm] (wq) {$v_{\overline{p}c}$};
    \draw (1,-1) node[index gray, minimum size=0.8cm] (w4) {$v_{\overline{pc}}$};
  	
  	\end{tikzpicture}
      \hspace{0.5cm}
  }
  \caption{Examples of the Semantics}
  \label{fig:examples}
  \end{figure}

Figure \ref{fig:examples} displays some examples of the semantics. Each node represents a valuation, and is labelled according to what is evaluated as true ($1$) by that valuation, e.g., $v_{p\overline{q}}(p)=1$ and  $v_{p\overline{q}}(q)=0$. The encircled area is a team we wish to highlight. The team $t$ in figure \ref{fig:a} satisfies the formula $p\vee \lnot p$ because it can be split into two subteams, one ($\{v_{p\overline{q}}\}$) satisfying the disjunct $p$, the other (the empty team $\emptyset$) satisfying the other disjunct $\lnot p$ (the connective $\vee$ is commonly known as the \emph{split disjunction}). With the nonemptiness atom $\NE$, we can enforce that teams be split into nonempty subteams: the team $s$ in Figure \ref{fig:b} satisfies the formula $(p\land \NE) \vee (\lnot p \land \NE)$ because it can be split into a nonempty team ($\{v_{pq}\}$) satisfying $p$, and a nonempty team ($\{v_{\overline{p}q}\}$) satisfying $\lnot p$.

Figure \ref{fig:c} relates to the applications of the atom $\NE$ in formal semantics and philosophical logic. On one possible interpretation of team semantics, teams represent \emph{information states}: if the true information that one possesses is represented by a team $t$, one knows that the valuation that represents the way that things actually are (the actual world) is among the valuations in $t$, and hence whatever holds in all valuations in $t$ must also be the case in the actual world.

Consider the defined operator $\Diamond \phi := (\phi \land \NE)\vee\top$---observe that the formula $\Diamond \phi$ is true in a team $t$ just in case $t$ contains a nonempty subteam in which $\phi$ is true. This operator is known as the \emph{might-operator}: it, along with other similar connectives, has been used to model the meanings of epistemic possibility modalities such as the `might' in ``It might be raining'' (see, e.g., \cite{veltman,yalcin}). For if one's information state contains a nonempty substate that supports the assertion that it is raining (if $t\models \Diamond r$), then the actual world might be one in which it is raining, and so, for all that one knows, it might be raining.

Now assume that Bob says, ``Sue went to the park or to the cinema.'' One natural inference from this might be that Bob deems it possible both that Sue went to the park and that she went to the cinema. The atom $\NE$ has been used to model inferences such as this one by Aloni and her collaborators (see, e.g., \cite{aloni2022,degano2025}). On Aloni's account, humans have a cognitive tendency to disregard structures that verify sentences by virtue of some empty configuration (the \emph{neglect-zero tendency}). In the setting of team semantics, these empty configurations are represented by the empty team, and the neglect-zero tendency is modelled by appending `$\land \NE$' to all subformulas of any formula that formalises a sentence that is being interpreted under this tendency. Then, for instance, the inference discussed above may be accounted for by pointing to the entailment $((p\land \NE)\vee (c\land \NE))\land \NE\models \Diamond p \land \Diamond c$---Bob's utterance, when interpreted under the tendency (we append `$\land \NE$' to all subformulas of the formula $p\vee c$, resulting in the formula $((p\land \NE)\vee (c\land \NE))\land \NE$) indicates that according to his information, Sue might have gone to the park and might have gone to the cinema ($\Diamond p \land \Diamond c$). See Figure \ref{fig:c}.\footnote{We have simplified and modified this account due to space constraints and to make it function in the propositional setting. For more on the actual account, see \cite{aloni2022,degano2025}.}

We define standard \emph{closure properties}. We say that $\phi$ has the
\begin{align*}
&\text{\textbf{empty team property}} &&\text{iff}  &&\emptyset \models \phi;\\
&\text{\textbf{downward closure property}
}  &&\text{iff} &&[t\models \phi\text{ and } s\subseteq t]\implies s\models \phi;\\
&\text{\textbf{union closure property}} &&\text{iff} &&[t\models \phi \text{ for all }t\in T\neq \emptyset]\implies \bigcup T\models \phi;\\
&\text{\textbf{flatness property}} &&\text{iff} &&  t\models \phi \iff [\{v\}\models\phi\text{ for all } v\in t];\\ 
&\text{\textbf{convexity property}
}  &&\text{iff} &&[t\models \phi, u\models \phi,\text{ and } u\subseteq s\subseteq t]\implies s\models \phi.
\end{align*}

It is easy to see that $\phi$ is flat iff it has the empty team, downward closure, and union closure properties.

\begin{proposition}[\cite{anttila2025convexteamlogics}] \label{cPLNE:prop:closure_props}
    All formulas of $\PLNE$ satisfy the union closure and convexity properties. All formulas of $\PL$ additionally satisfy the empty team property, and hence also the downward closure and flatness properties.
\end{proposition}

\begin{proposition}\label{cPLNE:prop:empty_team_property_iff_PL}
For $\phi\in \PLNE$, $\phi$ has the empty team property iff $\phi\in \PL$.
\end{proposition}
\begin{proof}
    The right-to-left direction follows by Proposition \ref{cPLNE:prop:closure_props}. For the left-to-right direction, an easy induction on the complexity of $\phi$ shows that if $\NE$ occurs in $\phi$, then $\phi$ does not have the empty team property. We only give the disjunction case here. If $\phi=\psi\vee\chi$ and $\NE$ occurs in $\phi$, then it must occur in $\psi$ or $\chi$. Without loss of generality, assume that $\NE$ occurs in $\psi$. Then, by the induction hypothesis, $\psi$ does not have the empty team property, but then clearly neither does $\psi\vee\chi$.
\end{proof}

We also have that the team semantics of formulas of $\PL$ on singletons coincide with their standard single-valuation semantics: $\{v\}\models \alpha \iff v\models \alpha$. Therefore:

\begin{proposition} \label{cPLNE:prop:CPL_flat_standard_semantics}
    For any $\alpha\in \PL$:
\[t\models \alpha \iff \{v\}\models \alpha \text{ for all }v\in t \iff v\models \alpha \text{ for all }v\in t  .\]
\end{proposition}

\begin{corollary}\label{cPLNE:coro:CPL_entailment_eq_team_entailment}
For any set $\Lambda\cup\{\alpha\}$ of classical formulas, $\Lambda\models^c\alpha \iff \Lambda\models\alpha$, where $\models^c$ stands for the usual (single-valuation) entailment relation.
\end{corollary}

\input{sections/mc}
\input{sections/sat}

\input{sections/val}

\section{Conclusion} \label{cPLNE:section:conclusion}
In this article, we have studied the complexity aspects of the logic $\PLNE$. In particular, we proved that its model-checking problem is in P; that its satisfiability problem is NP-complete; and that its validity problem is coNP-complete. In addition, we proved that $\PLNE$ has certain semantic properties that are interesting in their own right (in Lemmas \ref{cPLNE:lemma:satisfiability_in_small_team} and \ref{cPLNE:lemma:NE_validity}, and in Corollary \ref{cPLNE:coro:small_sat}).

We conclude by outlining some avenues for future research.

As we noted in Section \ref{cPLNE:section:introduction}, our results show that adding the nonemptiness atom $\NE$ to classical propositional logic does not increase the complexity of the satisfiability problem, or that of the validity problem. It would, then, be interesting to study the entailment problem of $\PLNE$, as well as to investigate whether our result concerning the model-checking problem can be improved upon (by either proving a stricter inclusion result, or a hardness result---cf. the P-hardness result for $\PL(\subseteq)$ in \cite{hella2019}), in order to determine whether the same holds for these two problems.

We also mentioned that there are no other results on the complexity of convex but not downward closed logics in the literature. The next natural step, then, would be to study the complexity of other such logics. We conjecture that many of our results would extend naturally to other convex union-closed logics such as the modal extensions of $\PLNE$ (see \cite{aloni2022,aloni2023,anttila2025convexteamlogics}). Analogues of the results involving the use of convexity (such as Lemma \ref{cPLNE:lemma:satisfiability_in_small_team}) may also hold for logics which are convex but not union-closed (see \cite{anttila2025convexteamlogics}).

\begin{credits}
\subsubsection{\ackname} The first author received support from the Academy of Finland, decision number 368671, and from the European Research Council (ERC), grant 101020762. The authors would like to thank Maria Aloni for helpful discussions related to the content of this article; as well as the anonymous reviewers, for their valuable comments and suggestions.

\subsubsection{\discintname}
The authors have no competing interests to declare that are
relevant to the content of this article.
\end{credits}

\bibliographystyle{splncs04}
\bibliography{bibb}

\end{document}

%% file: sections/mc.tex
\section{Model-Checking Problem} \label{cPLNE:section:MC}

The model-checking problem $\MC{\mathrm{L}}$ for a team-based logic $\mathrm{L}$ is defined as follows: given a formula $\phi$ of $\mathrm{L}$ and a team $t$, determine whether $t\models\phi$. In this section, we show that $\MC{\PLNE}\in \mathrm{P}$. This is proved using an argument that is similar to that used for propositional inclusion logic $\PL(\subseteq)$ in \cite{hella2019}, but we make some modifications to account for the fact that the empty team property fails in $\PLNE$ (whereas it holds in $\PL(\subseteq)$).

Let $\mathrm{maxsub}(t,\phi)$ denote the maximum (with respect to the subset relation) subteam $s$ of $t$ such that $s\models \phi$. Note that if $\phi$ is union closed and has the empty team property, such a maximum subteam always exists---this is the case in particular for $\phi \in\PL(\subseteq)$. This is not the case, however, for formulas of $\PLNE$: the formula $\NE$, for instance, is not satisfiable in any subteam of the empty team. We do, however, have the following property:
\begin{lemma} \label{cPLNE:lemma:union_closed_maxsub}
    For union-closed $\phi$, if $\phi$ is satisfiable in some subteam of $t$, then $\mathrm{maxsub}(t,\phi)$ exists.
\end{lemma}
\begin{proof}
    Clearly $\mathrm{maxsub}(t,\phi)=\bigcup \{s\subseteq t\mid s\models \phi\}$.
\end{proof}
Let $\mathrm{Maxsub}(t,\phi)=\mathrm{maxsub}(t,\phi)$ if $\mathrm{maxsub}(t,\phi)$ exists, and $\mathrm{Maxsub}(t,\phi)=\mathbf{0}$ otherwise, where $\mathbf{0}$ is a ``null team'' such that $\mathbf{0}\not\models \phi$ for all $\phi$. (Note that $\mathbf{0}$ is not an actual team, and in particular is not the empty team $\emptyset$---it is, rather, a notational shorthand we use to indicate that a formula $\phi$ is not satisfiable in any subteam of a team $t$.) Let $|t|$ denote the size of $t$, and $|\phi|$ the number of symbols in $\phi$.

\begin{lemma} \label{cPLNE:lemma:maxsub}
    If $\phi$ is a literal, then $\mathrm{Maxsub}(t,\phi)$ can be computed in polynomial time with respect to $|t|+|\phi|$.
\end{lemma}
\begin{proof}
    First check whether $\phi=\NE$. If yes, then if $t=\emptyset$, return $\mathbf{0}$, and otherwise return $t$. Clearly this can be done in polynomial time with respect to $|t|+|\phi|$. If $\phi\neq \NE$, check $\mathrm{maxsub}(t,\phi)$; the result then follows by \cite[Lemma 5.1]{hella2019}.
\end{proof}

\begin{proposition}\label{cPLNE:prop:MC(PLNE)_in_P}
    $\MC{\PLNE}\in \mathrm{P}$.
\end{proposition}
\begin{proof}
We want to check whether $t\models \phi$ in polynomial time w.r.t. $|t|+|\phi|$.

    Let $\mathrm{subOcc(\phi)}$ denote the set of all occurrences of subformulas of $\phi$. In the following, we denote occurrences as if they were formulas.

    A function $f:\mathrm{subOcc}(\phi)\to  2^{2^\mathsf{P(\phi)}}\cup \{\mathbf{0}\}$ is called a \emph{labelling function}. We describe an algorithm for computing a sequence $f_0,f_1,\ldots$ of such labelling functions. In what follows, we let $\mathbf{0}\cap s:=\mathbf{0}$ and $\mathbf{0}\cup s:=\mathbf{0}$ for all teams $s$.

    \begin{itemize}
        \item $f_0(\psi)=t$ for all $\psi\in\mathrm{subOcc(\phi)} $.
        \item For odd $i\in \mathbb{N}$, define $f_i(\psi)$ bottom-up as follows:
        \begin{enumerate}
            \item For literal $\psi$, define $f_i(\psi):=\begin{cases}
                \mathrm{Maxsub}(f_{i-1}(\psi),\psi) &\text{ if }f_{i-1}\neq \mathbf{0};\\
                \mathbf{0} &\text{ if }f_{i-1}= \mathbf{0}.
            \end{cases}$.
            \item $f_{i}(\psi\land \chi):= f_{i}(\psi)\cap f_i(\chi)$.
            \item $f_{i}(\psi\vee \chi):=f_{i}(\psi)\cup f_i(\chi)$.
        \end{enumerate}
                \item For even $i\in \mathbb{N}$ larger than $0$, define $f_i(\psi)$ top-down as follows:
        \begin{enumerate} 
        \item $f_i(\phi):=f_{i-1}(\phi)$.

            \item If $\psi=\chi\land\theta$, define $f_i(\chi):=f_i(\theta):=f_i(\psi)$.
         
            \item If $\psi=\chi\vee\theta$, define $f_i(\chi):=f_{i-1}(\chi)\cap f_i(\chi\vee\theta)$ and $f_i(\theta):=f_{i-1}(\theta)\cap f_i(\chi\vee\theta)$.
        \end{enumerate}
    \end{itemize}

    For a given $i\in \mathbb{N}$, an easy induction on the complexity of formulas (with the induction going from more complex formulas to less complex for even $i\in \mathbb{N}$) shows that $f_{i+1}(\psi)\subseteq f_{i}(\psi)$, where we stipulate that $\mathbf{0}\subseteq s$ and $\mathbf{0}\subseteq \mathbf{0}$ (and $s\not \subseteq \mathbf{0}$) for all teams $s$.  It follows that there is an integer $j\leq 2\cdot(|t|+1)\cdot |\phi|$ such that $f_{j+2}=f_{j+1}=f_j$. 
    We denote this fixed point of the sequence $f_0,f_1,\ldots$ by $f_\infty$. By Lemma \ref{cPLNE:lemma:maxsub}, the outcome of $\mathrm{Maxsub}$ is computable in polynomial time with respect to its input. Given this, $f_{i+1}$ can be computed in polynomial time from $f_i$ with respect to $|t|+|\phi|$, so that $f_\infty$ can also be computed in polynomial time with respect to $|t|+|\phi|$.\\

\noindent
    \textbf{Claim 1.} For each $\psi\in \mathrm{subOcc}(\phi)$, $f_\infty(\psi)\models \psi$ iff $f_\infty(\psi)\neq \mathbf{0}$.\\

\noindent
    \emph{Proof of Claim 1.} The left-to-right direction is immediate. We prove the right-to-left direction by induction on $\psi$. Fix an odd $i$ and even $j$ such that $f_\infty=f_i=f_j$.
    \begin{itemize}
        \item If $\psi$ is a literal, the claim is true because $f_\infty=f_i$, and if $f_\infty(\psi)=f_i(\psi)\neq \mathbf{0}$, then $f_\infty(\psi)=  f_i(\psi)=\mathrm{maxsub}(f_{i-1}(\psi),\psi)$, where $f_{i-1}(\psi)\neq \mathbf{0}$.
        \item Assume $\psi=\chi \land \theta$, and that the claim is true for $\chi$ and $\theta$. Since $f_\infty=f_j$, we have $f_\infty(\psi)=f_\infty(\chi)=f_\infty(\theta)$, so clearly by the induction hypothesis, $f_\infty(\psi)\neq \mathbf{0}$ implies $f_\infty(\psi)\models \chi \land \theta$.
        \item Assume $\psi=\chi \vee \theta$, and that the claim is true for $\chi$ and $\theta$. Since $f_\infty=f_i$, we have $f_\infty(\psi)=f_\infty(\chi)\cup f_\infty(\theta)$. By the induction hypothesis, $f_\infty(\chi)\neq \mathbf{0}$ implies $f_\infty(\chi)\models \chi $, and $f_\infty(\theta)\neq \mathbf{0}$ implies $f_\infty(\theta)\models \chi $. We show that $f_\infty(\psi)\neq \mathbf{0}$ implies $f_\infty(\psi)\models \chi \lor \theta$. If $f_\infty(\psi)=f_\infty(\chi)\cup f_\infty(\theta)\not\models \chi \lor \theta$, then $f_\infty(\chi)\not\models \chi $ or $f_\infty(\theta)\not\models \theta $, whence $f_\infty(\chi)= \mathbf{0}$ or $f_\infty(\theta)= \mathbf{0}$, either of which imply $f_\infty(\psi)= \mathbf{0}$.
    \end{itemize}

    In particular, (1) if $f_\infty(\phi)=s$ (for some actual non-null team $s$), then $s\models \phi$. To finish the proof, it therefore suffices to prove that (2) if $t\models \phi$, then $f_\infty(\phi)=t$: for then, to check whether $t\models \phi$, check whether $t=f_\infty(\phi)$. If yes, then $f_\infty(\phi)=t\neq \mathbf{0}$ so by (1), $f_\infty(\phi)=t\models \phi$; if no, then by (2), $t\not\models \phi$.
    
    To prove (2), suppose $t\models \phi$. It is easy to see that then, for each $\psi\in \mathrm{subOcc}(\phi)$, there is a team $t_\psi\subseteq t$ such that $t_\psi\models \psi$ and
    \begin{itemize}
        \item $t_\phi=t$;
        \item $t_{\chi\land\theta}=t_\chi=t_\theta$;
        \item $t_{\chi\vee\theta}=t_\chi \cup t_\theta$.
    \end{itemize}

\noindent
 \textbf{Claim 2.} $t_\psi\subseteq f_i(\psi)$ for all $\psi\in  \mathrm{subOcc}(\phi)$ and all $i \in \mathbb{N}$.\\

    \noindent
    \emph{Proof of Claim 2.} By induction on $i$. For $i=0$ this is obvious since $f_0(\psi)=t$ for all $\psi$. Now assume $i+1$ is odd and that the claim is true for $i$. We prove the claim by a subinduction on $\psi$.
    \begin{itemize}
        \item If $\psi$ is a literal, then since by the induction hypothesis $t_\psi\subseteq f_i(\psi)$, we have $f_i(\psi)\neq \mathbf{0}$, whence $ f_{i+1}(\psi)=\mathrm{Maxsub}(f_i(\psi),\psi)$.  Since $t_\psi\models \psi$ and $t_\psi\subseteq f_i(\psi)$, we have by Lemma \ref{cPLNE:lemma:union_closed_maxsub} that $f_{i+1}(\psi)=\mathrm{Maxsub}(f_i(\psi),\psi)\neq \mathbf{0}$ so $f_{i+1}(\psi)=\mathrm{maxsub}(f_i(\psi),\psi)$; therefore $t_\psi \subseteq f_{i+1}(\psi) $.
        \item Assume that $\psi=\chi \land \theta$. By the subinduction hypothesis, $t_\psi=t_\chi\subseteq f_{i+1}(\chi)$ and $t_\psi=t_\theta\subseteq f_{i+1}(\theta)$, and so $t_\psi \subseteq f_{i+1}(\chi)\cap f_{i+1}(\theta)=f_{i+1}(\psi)$. 
         \item Assume that $\psi=\chi \lor \theta$. By the subinduction hypothesis, $t_\chi\subseteq f_{i+1}(\chi)$ and $t_\theta\subseteq f_{i+1}(\theta)$. Then $t_\psi=t_\chi\cup t_\theta \subseteq f_{i+1}(\chi)\cup f_{i+1}(\theta)=f_{i+1}(\psi)$.
    \end{itemize}
    Now assume $i+1$ is even and that the claim is true for $i$. We prove the claim by a top-to-bottom subinduction on $\psi$.
    \begin{itemize}
        \item For $\psi=\phi$, by the induction hypothesis, $t_\psi\subseteq f_i(\psi)=f_{i+1}(\psi)$.
        \item Assume that $\psi=\chi \land \theta$. By the subinduction hypothesis, $t_\chi=t_\theta=t_\psi\subseteq f_{i+1}(\psi)=f_{i+1}(\chi)=f_{i+1}(\theta)$.
        \item Assume that $\psi=\chi \lor \theta$. By the subinduction hypothesis, $t_\psi\subseteq f_{i+1}(\psi)$. By the induction hypothesis, $t_\chi\subseteq f_{i}(\chi)$. Then since $t_\chi \subseteq t_\psi$, we have $t_\chi \subseteq f_{i}(\chi)\cap t_\psi \subseteq f_{i}(\chi)\cap f_{i+1}(\psi)=f_{i+1}(\chi) $. Similarly $t_\theta \subseteq f_{i+1}(\theta) $.
    \end{itemize}

    Given Claim 2, $t=t_\phi\subseteq f_\infty(\phi)$. Since $f_\infty(\phi)\subseteq f_0(\phi)=t$, we have $f_\infty(\phi)=t$.
    \end{proof}
    
    In the proof of Proposition \ref{cPLNE:prop:MC(PLNE)_in_P}, we have followed the proof for $\MC{\PL(\subseteq)}\in \mathrm{P}$ in \cite{hella2019} in defining a sequence of labelling functions $f_i$ for all $i\in \mathbb{N}$, and finding a fixed point $f_\infty$ of this sequence. Presenting the proof in this way allows for an elegant way to prove the required facts about the relevant labelled subteams of the target team  $t$, as well as for an easy comparison with the proof in \cite{hella2019}. It appears, though, that for $\PLNE$, the sequence of functions reaches its fixed point only a few elements into this sequence---seemingly by $i=4$---whereas this need not be the case for $\PL(\subseteq)$. We leave a detailed investigation into this matter, as well as any possible improvements to the algorithm, for future work.

%% file: sections/sat.tex
\section{Satisfiability Problem} \label{cPLNE:section:SAT}

We follow \cite{hannula2018} in defining the satisfiability problem $\SAT{\mathrm{L}}$ for a team-based logic $\mathrm{L}$ as follows: given a formula $\phi$ of $\mathrm{L}$, determine whether there exists some nonempty team $t$ such that $t\models\phi$ (satisfaction in the empty team is easy to decide, both for \PLNE---see Proposition \ref{cPLNE:prop:empty_team_property_iff_PL}---as well as for many other team-based logics---see \cite{hannula2018}). In this section, we show that $\SAT{\PLNE}$ is NP-complete. Hardness follows immediately by the NP-hardness of $\SAT{\PL}$ together with Proposition \ref{cPLNE:prop:CPL_flat_standard_semantics}, so we focus on showing $\SAT{\PLNE}\in \mathrm{NP}$. 

\begin{lemma}\label{cPLNE:lemma:satisfiability_in_small_team}
    For $\phi \in \PLNE$, if $t\models\phi$, then there exists $s\subseteq t$ such that $s\models \phi$ and $|s|\leq |\phi|_{\NE}$, where $|\phi|_{\NE}$ is the number of occurrences of $\NE$ in $\phi$.
\end{lemma}
\begin{proof}
    By induction on the complexity of $\phi$.
    \begin{itemize}
        \item $\phi$ is a classical literal. Immediate by the empty team property.
        \item $\phi=\NE$. Take any singleton subteam of $t$.
        \item $\phi=\psi\vee\chi$. If $t\models \psi\vee\chi$, we have $t=s\cup u$, where $s\models\psi$ and $u\models \chi$. By the induction hypothesis, there are $s'\subseteq s$ and $u'\subseteq u$ such that $s'\models \psi$, $u'\models \chi$,  $|s'|\leq |\psi|_{\NE}$, and $|u'|\leq |\chi|_{\NE}$. Then $s'\cup u'\subseteq s\cup u\subseteq t$, $s'\cup u'\models \psi\vee\chi$, and $|s'\cup u'|\leq |s'|+|u'| \leq |\psi|_{\NE}+|\chi|_{\NE} =  |\psi\vee\chi|_{\NE}$.
        \item $\phi=\psi\land \chi$.  If $t\models \psi\land \chi$, by the induction hypothesis there are $s,u\subseteq t$ such that $s\models \psi$ and $u\models \chi$. Then since $s,u\subseteq s\cup u\subseteq t$, we have by convexity that $s\cup u \models \psi \land \chi$. And we have $|s\cup u|\leq |s|+|u|\leq  |\psi|_{\NE}+|\chi|_{\NE} = |\psi\land\chi|_{\NE}$
    \end{itemize}
\end{proof}

    \begin{corollary}[Small model property] \label{cPLNE:coro:small_sat}
        If $\phi \in \PLNE$ is satisfiable, it is satisfiable in a team of size $\leq |\phi|_{\NE}$.
    \end{corollary}

    \begin{proposition} \label{cPLNE:prop:SAT(PLNE)_in_NP}
        $\SAT{\PLNE}\in \mathrm{NP}$.
    \end{proposition}
    \begin{proof}
    First check whether $\phi\in \PL$. If yes, the result follows by Proposition \ref{cPLNE:prop:CPL_flat_standard_semantics} and $\SAT{\PL}\in \mathrm{NP}$. If no, note that we have $|\phi|_{\NE}\geq 1$. The required algorithm  guesses a nonempty team of size at most $|\phi|_{\NE}$ and then runs the deterministic poly-time model-checking algorithm (Proposition \ref{cPLNE:prop:MC(PLNE)_in_P}). It is easy to see that the running time of this procedure is polynomial in $|\phi|$, and hence containment in NP follows.
        \end{proof}
        \begin{corollary}\label{cPLNE:coro:SAT(PLNE)_NP_complete}
            $\SAT{\PLNE}$ is $\mathrm{NP}$-complete.
        \end{corollary}

%% file: sections/val.tex
\section{Validity Problem}\label{cPLNE:section:VAL}

We follow \cite{hannula2018} in defining the validity problem $\VAL{\mathrm{L}}$ for a team-based logic $\mathrm{L}$ as follows: given a formula $\phi$ of $\mathrm{L}$, determine whether $t\models \phi$ for all nonempty teams $\phi$ with domain $\supseteq \mathsf{P}(\phi)$---i.e., whether $\NE\models \phi$ (note that due to Proposition \ref{cPLNE:prop:empty_team_property_iff_PL} and Corollary \ref{cPLNE:coro:CPL_entailment_eq_team_entailment}, the validity problem of $\PLNE$ defined with respect to all teams---that is, including the empty team---is easily reducible to the validity problem of $\PL$, and vice versa). In this section, we show that $\VAL{\PLNE}$ is $\mathrm{coNP}$-complete.

 We make use of the following notion: the \emph{flattening} $\phi^f\in \PL$ of $\phi\in \PLNE$ is defined by $\phi^f:=\phi(\top/\NE)$ (see \cite{anttila2024}; cf. \cite{vaananen2007}). It is easy to see that $\phi \models \phi^f$; observe also that by Proposition \ref{cPLNE:prop:empty_team_property_iff_PL}, $\phi\in \PLNE$ has the empty team property iff $\phi=\phi^f$.

 \begin{lemma} \label{cPLNE:lemma:NE_validity} For any $\phi\in \PLNE$ without the empty team property, the following are equivalent:
 \begin{enumerate}
     \item $\NE\models \phi$;
     \item for all $\psi\in \mathrm{subOcc}(\phi)$ without the empty team property, $\NE\models \psi$;
     \item for all $\psi\in \mathrm{subOcc}(\phi)$ without the empty team property, $\models \psi^f$.
 \end{enumerate}
    \end{lemma}
    \begin{proof} 
    $1\implies 2$: Suppose $\NE\models \phi$. Fix some nonempty team $t$; we will show that for all $\psi\in \mathrm{subOcc}(\phi)$ without the empty team property, $t\models \psi$. We have that for any $v\in t$, $\{v\}\models \phi$. Then for each $\psi\in\mathrm{subOcc}(\phi)$, there is a subteam $\{v\}_\psi\subseteq \{v\}$ as described in the proof of Proposition \ref{cPLNE:prop:MC(PLNE)_in_P}. For any $\psi\in\mathrm{subOcc}(\phi)$ without the empty team property, we must have $\{v\}_\psi\neq \emptyset$, so $\{v\}_\psi=\{v\}$, whence $\{v\}\models \psi$. Since $v$ was arbitrary, by union closure we have $t\models \psi$.

    $2\implies 3$: By $\psi\models \psi^f$ and the empty team property of $\psi^f$.

    $3\implies 2$: Suppose that for all $\psi\in \mathrm{subOcc}(\phi)$ without the empty team property we have $\models \psi^f$. We show that for all $\psi\in \mathrm{subOcc}(\phi)$, either $\psi$ has the empty team property or $\NE \models \psi$ by induction on the structure of $\psi$.
        
        \begin{itemize}
            \item[-]  For $\psi$ a classical literal, $\psi$ has the empty team property by Proposition \ref{cPLNE:prop:closure_props}.
            \item[-] For $\psi=\NE$, the result is immediate.
            \item[-] For $\psi=\chi \land \theta$, if both $\chi$ and $\theta$ have the empty team property, then so does $\psi$, and we are done. If neither has the empty team property, then neither does $\chi\land \theta$. Then by the induction hypothesis, $\NE\models \chi$ and $\NE\models \theta$, whence $\NE\models \chi \land \theta$. If $\chi$ has the empty team property and $\theta$ does not, then $\chi \land \theta$ does not. Therefore, by assumption, $\models (\chi\land \theta)^f$, whence $\models \chi^f$ and $\models\theta^f$. Since $\chi$ has the empty team property, we have $\chi^f=\chi$, so $\models \chi$, whence also $\NE\models\chi$. And since $\theta$ does not have the empty team property, we have $\NE\models \theta$ by the induction hypothesis, so that $\NE\models \chi \land \theta$. Similarly for the final case.
            \item[-] For $\psi=\chi \vee \theta$, if both $\chi $ and $\theta$ have the empty team property, then so does $\chi\vee\theta$. If neither does, then neither does $\chi\vee\theta$; we have $\NE\models \chi$ and $\NE\models \theta$ by the induction hypothesis, and so $\NE\models \chi \vee\theta $. If $\chi$ does and $\theta$ does not, then we have $\NE\models \theta$ by the induction hypothesis. Then also $\NE\models \chi \vee\theta $ since for any team $t$ we have $t=\emptyset \cup t$. Similarly for the final case.
        \end{itemize}
            $2\implies 1$: Immediate.
    \end{proof}

    \begin{proposition} \label{cPLNE:prop:VAL(PLNE)_coNP_complete}
        $\mathrm{VAL}[\PLNE]$ is $\mathrm{coNP}$-complete.
    \end{proposition}
    \begin{proof}
        Inclusion: First check whether $\phi \in \PL$. If yes, then by Corollary \ref{cPLNE:coro:CPL_entailment_eq_team_entailment} and the empty team property, we have $\NE\models \phi $ iff $\models \phi$ iff $\models^c \phi$, so it suffices to check whether $\models^c \phi$---the result then follows by $\VAL{\PL}\in \mathrm{coNP}$. If no, then by Proposition \ref{cPLNE:prop:empty_team_property_iff_PL}, $\phi$ does not have the empty team property, so by Lemma \ref{cPLNE:lemma:NE_validity} and Corollary \ref{cPLNE:coro:CPL_entailment_eq_team_entailment}, it suffices to check whether $\models^c \psi^f$ for all $\psi\in \mathrm{subOcc}(\phi)$. This problem can be shown to be in $\mathrm{coNP}$ by showing that its complement is in NP. For this, we non-deterministically guess a formula $\psi\in \mathrm{subOcc}(\phi)$ and then accept if the (classical) negation of $\psi^f$ is satisfiable (given the standard Tarskian semantics). This is in $\mathrm{NP}$ given that $\SAT{\PL}\in \mathrm{NP}$, and so $\mathrm{VAL}[\PLNE]\in \mathrm{coNP} $.

        Hardness: We reduce $\VAL{\PL}$ to $\mathrm{VAL}[\PLNE]$. Fix $\phi\in \PL$. Given Corollary \ref{cPLNE:coro:CPL_entailment_eq_team_entailment} and the empty team property, $\models^c \phi$ iff $\models \phi$ iff $\NE\models \phi $.
    \end{proof}